\documentclass[sigconf]{acmart}

\usepackage[figuresright]{rotating}
\usepackage{subfigure}
\usepackage{mathtools}
\usepackage{tablefootnote}
\usepackage{threeparttable}
\usepackage{graphics} 
\usepackage{amsthm}
\usepackage[english]{babel}

\newtheorem{prop}{Proposition}
\usepackage{algorithm}
\usepackage{algorithmic}
\AtBeginDocument{%
  \providecommand\BibTeX{{%
    \normalfont B\kern-0.5em{\scshape i\kern-0.25em b}\kern-0.8em\TeX}}}

\begin{document}

\title{Learning Computation Bounds for Branch-and-Bound Algorithms to k-plex Extraction}


\author{Yun-Ya Huang}
\affiliation{%
  \institution{National Tsing Hua University}
  \city{Hsinchu}
  \country{Taiwan}}
\email{ilcilc1975@gmail.com}

\author{Chih-Ya Shen}
\affiliation{%
  \institution{National Tsing Hua University}
  \city{Hsinchu}
  \country{Taiwan}}
\email{shenchihya@gmail.com}


\begin{abstract}
  k-plex is a representative definition of communities in networks. While the cliques is too stiff to applicate to real cases, the k-plex relaxes the notion of the clique, allowing each node to miss up to k connections. Although k-plexes are more flexible than cliques, finding them is more challenging as their number is greater. In this paper, we aim to detect the k-plex under the size and time constraints, leveraging the new vision of automated learning bounding strategy. We introduce the constraint learning concept to learn the bound strategy from the branch and bound process and develop it into a Mixed Integer Programming framework. While most of the work is dedicated on learn the branch strategy in branch and bound-based algorithms, we focus on the learn to bound strategy which needs to handle the problem that learned strategy might not examine the feasible solution. We adopted the MILP framework and design a set of variables relative to the k-plex property as our constraint space to learn the strategy. The learn to bound strategy learning the original strategy function also reduces the computation load of the bound process to accelerate the branch and bound algorithm. Note that the learn to bound concept can apply to any branch and bound based algorithm with the appropriate framework. We conduct the experiment on different networks, the results show that our learn to branch and bound method does accelerate the original branch and bound method and outperforms other baselines, while also being able to generalize on different graph properties.
\end{abstract}

%

\ccsdesc[500]{Mathematics of computing~Graph algorithms}
\ccsdesc[500]{Mathematics of computing~Combinatorial optimization}
\ccsdesc[500]{Theory of computation~Constraint and logic programming}
\ccsdesc[500]{Theory of computation~Graph algorithms analysis}
\keywords{k-plex, constraint Learning, branch and bound, mixed integer linear programming}

\maketitle

\section{Introduction}
\noindent Most of the real-world networks represent the relationship in different scenarios. The edge distribution of these networks is not uniform, and it is obvious that the set of nodes highly connected to each other, often called \emph{communities}, is quite different from the other nodes in the network. Detecting the \emph{communities} can help to explore the fundamental properties of the large networks, thus the automated detecting communities is an important issue and has been largely investigated \cite{fortunato2010community}.

A clique is a set of nodes in a network that every two distinct vertices in the set are connected to each other. Since the property of clique is too stiff to apply in practice \cite{pattillo2012clique}, a more appropriate property k-plex is able to represent communities and also suitable for real scenarios. k-plex is each vertex of the induced subgraph is connected to at least m-k other vertices, where m is the number of vertices in the induced subgraph. In other words, each vertex has edges with all the others, with the possible exception of up to k missing neighbors (including itself). Thus, k-plexes are a simple and intuitive generalization of cliques.

The problem of finding k-plex has emerged in social network analysis \cite{balasundaram2011clique}, and it also be applied to some important fields such as employing graph-based data mining \cite{berlowitz2015efficient}, \cite{pattillo2012clique}, \cite{zhai2016fast}. Unfortunately, detecting k-plexes in a network is more difficult than detecting k-cliques since the maximal k-plexes are much more numerous than maximal cliques, and most of the efficient algorithms for computing maximal k-plexes can only be used on small-size graphs and also required a lot of computation. Thus we want to develop an effective algorithm for detecting k-plex with certain quality in a short time.

There are lots of works dedicated to detecting k-plex problem \cite{berlowitz2015efficient}, \cite{conte2017fast}, \cite{zhou2021improving}, \cite{mcclosky2012combinatorial}, \cite{jiang2021new}, \cite{wang2017query}. Part of them use the branch and bound technique and achieves a good performance on finding maximum k-plex \cite{zhou2021improving},\cite{mcclosky2012combinatorial}, \cite{wang2017query}. Although the branch and bound based algorithm successfully work on detecting k-plex, the common issue of the branch and bound algorithm is that it requires a lot of effort to design the heuristic strategy and the bound strategy, and it is annoying to build a strategy that improves the power of the bound conditions also reduces the computation time by handcrafting. In this work, we make the first attempt to use the concept of constraint learning to automated learning bound strategy from the branch and bound process which helps the branch and bound algorithm work effectively and detect the k-plex quickly.

Using constraint learning to effectively detect k-plex needs to due with the following difficulty: 1) it is required to propose the overall framework and learning flow including deciding which kind of information as an example data and the features (variable) of the examples and what we want the constraint learning to distinguish; 2) the variable chosen need to be relative to the target problem and can be generalized on it; 3) the computation load of the variable need to take into the consideration; 4) the complexity of the process function of variables needs to be moderate that able to distinguish the examples data also not complicate to increase the MILP solvers loading. We design a complete framework that considers all the conditions mentioned above which lead to mining bounding strategy to efficiently detect the k-plex.

Although our work is dedicated to using the learn to bound idea to accelerate the process of k-plex extraction, the concept of constraint learning can apply to different problems with branch and bound based algorithms. Given the target problem and any basic bound strategy, we can perform the same framework we proposed on it with a suitable set of variables. In the branch and bound based algorithm, we want the bound can discard the exploration of infeasible solution precisely, this may require the complex variables that are able to describe the bound behavior well; on the other hand, we want to reduce the computation load of bound strategy, this required the achieving of the variable do not cost a lot of computation. For example, to calculate the neighbor of each vertex in the set need to traverse each vertex and edge, while the size of the set needs no computation. There seems to be a trade-off between the expressive power and the computation load of variables. Although the constraint learning framework can perform the learn to bound technique on different problems also different bound strategies, it requires the observation of the features of the target problem and corresponding bound strategy to design a suitable set of variables, leading to the efficient and powerful learned strategy.

Recently, machine learning has been adopted to solve the combinatorial optimization problems that have been mainly solved by algorithm approaches for the past few decades. Several works tackle the combinatorial optimization problems as a sequence transformation problem, \cite{li2018combinatorial}, \cite{gu2020pointer}. And several works demonstrated the problem can solve in by Reinforcement learning(RL) approach ,\cite{khalil2017learning}, \cite{barrett2020exploratory}, \cite{li2021deep}. They demonstrated RL can effectively solve graph optimization problems, learns to explore the potential combination, and exploit the experience without the requirement of the optimal solution as supervision. Different from the works that put their effort into designing the appropriate machine learning framework and then building the solution with lots of training and computation load. Our method applies constraint learning automatically learning to prevent the branch and bound based algorithm from examining the infeasible solution. We develop a framework that learns to bound for the branch and bound algorithm to accelerate it. 

There are also some works using machine learning to help accelerate the branch and bound algorithm, \cite{balcan2018learning}, \cite{zarpellon2021parameterizing}, \cite{qu2022improved} they put their effort into learning the branch strategy through different machine learning frameworks, and design the architecture of the learning process from MILP solver. The problem of learning branch strategy is relatively simple than learning to bound. The learning to branch reduces the size of the search tree most of the time, it may also fail in some conditions, but it is not a big issue because it only affects the efficiency but does not influence if the algorithm can find a feasible solution. In contrast, the problem of learning to bound needs to be due carefully in case of filtering out feasible solutions, which makes the whole algorithm end up generating solutions with low quality. Although learn to bound are a challenging issue, we consider that an effective bound process has a great potential of reducing the overall search process. If we can correctly bound the infeasible solution at an early stage, it may reduce a large amount of computation of branch and bound decision. This effect might be much more powerful than learn to branch since learn to branch only deal with the problem of clever select strategy, it might lead to a high quality of local solution but not powerful to reduce the overall search space. Thus, we provide a new framework to improve the efficiency of the bound strategy, expecting to accelerate the branch and bound search process.

In this work, we tackle the core problem of the branch and bound algorithm: the bound strategy for efficiently detecting k-plex in the networks. First, we introduce the simple branch and bound algorithm. Then we propose a framework in constraint learning concept to learn the bound strategy with the basic branch and bound algorithm searching process. Then learn to bound algorithm can learn the behavior of the original bound strategy with the branch and bound searching process on the small graphs in a short time, and the learned strategy does accelerate the algorithm with the low computation load in deciding whether examine the current searching states. 

The contributions are summarized as follows.
\begin{itemize}
  \item We present techniques to detect the k-plex of a network under the detecting size and time constraints. Our approach is based on the branch and bound algorithm and accelerates its efficiency with a complete framework of learning the bound strategy.
  
  \item We propose a constraint learning framework that automatically learns the bound strategy with the designing variable to reduce the computation load of deciding whether examine the current states. The learning framework in MILP successfully constructs the bound strategy with the searching process on the small graphs and then able to generalize on the real networks, which leads to improving the efficiency of the branch and bounce base algorithm.
  
  \item Experimental results on multiple benchmark datasets show that our proposed approaches outperform the state-of-the-art baselines. Our method is able to generalize on networks with different properties and perform stable with various parameters.
  
\end{itemize}

The rest of this dissertation is organized as follows. Section 2 will discuss related work on several related field. Section 3 is the algorithm design. We will introduce the proposed algorithm of learning to bound. Section 4 is the experimental result and Section 5 is the conclusion.

\section{Related Works}
\subsection{K-plex Detection}
The problem of detecting k-plex can be expressed in finding the maximum k-plex and enumerating the k-plex in graphs. Recent work like Berlowitz, Cohen, and Kimelfeld et al. produce efficient algorithms for the enumeration of maximal k-plexes and maximal connected k-plexes, the algorithm reduces the problem of enumerating k-plex for an arbitrary graph G to enumerating k-plex for the graph that almost satisfies k-plex \cite{berlowitz2015efficient}. 
Conte et al. are dedicated to reducing the large graph to a relatively small graph by the \emph{coreness} and the \emph{cliquness} property, which makes the existing k-plex enumerating algorithms possible to execute on large graphs \cite {conte2017fast}. 
Zhou et al. develop a branch and bound algorithm Maplex, which efficiently removes redundant vertices and edges, and used the graph color heuristic to obtain a tight upper bound for the exact branch-and-bound \cite {zhou2021improving}. 
Miao et al. propose a greedy randomized adaptive search procedure (GRASP) to overcomes the drawbacks of existing construction heuristics found for smaller values of parameter k are sometimes not found for larger k even though they are feasible. And the propose algorithm is able to detect a maximum k-plex in various networks \cite {miao2017approaches}.
Chen et al. inspired by the multi-armed bandit (MAB) problem in reinforcement learning (RL) to proposed two heuristics, BLP and DTCC, for the maximum k-plex problem. Then develop a local search algorithm named BDCC and improve it by a hyper heuristic strategy \cite {chen2019combining}. 
Although these works provide the method for searching k-plex from different perspectives, detecting k-plex problem still need a tremendous of time. Hence, we concentrated on detecting k-plex in a short time with the lower bound constraints to provide the real-time requirement.

\subsection{Combinatorial Optimization Problems}
In recent years, machine learning has improved the field of solving combinatorial optimization problems.
Khalil et al. automate the process of learning heuristics for the combinatorial optimization problems by a combination of reinforcement learning and graph embedding, successfully applied to optimization problems over graphs \cite {khalil2017learning}.
Li et al. propose the framework combines deep learning and algorithmic ideas, the graph convolution network is trained to estimate if the vertex in the graph is part of the optimal solution, and able to produce multiple solutions following the procedure to explore it \cite {li2018combinatorial}.
Barrett et al. propose ECO-DQN, a reinforcement learning based algorithm for the Max-Cut problem. Instead of the previous work that earlier decisions are not revisable, their agent seeks to continuously improve the solution by learning to explore at test time and generalizes well to unseen graph sizes and structures \cite {barrett2020exploratory}.
Karalias et al. presented an unsupervised learning framework for solving constrained combinatorial problems on graphs that utilizes a neural network to parametrize a probability distribution over sets to guarantee the quality of its solutions. They demonstrate this approach is able to obtain valid solutions to the maximum clique problem and local graph clustering \cite {karalias2020erdos}.
Li et al. propose a deep learning approach to approximately solves the covering salesman problem(CSP). The model captures the structural patterns and forms a dynamic embedding to handle the dynamic patterns of the problem. The model is trained using the reinforce algorithm and shows desirable properties of fast solving speed and the ability to generalize to unseen instances \cite {li2021deep}. 
Unlike the machine learning approaches, Our method uses the constraint learning concept for learning bound strategy which gives a new perspective on solving combinatorial optimization problems.

\subsection{Constraint Learning}
Recent work for mining constraints from data inherited the structure of the inductive logic programming and further expressed the problem to the mixed integer linear problem form, also performing their different synthesis strategy.
Pawlak and Krawiec et al. propose a method that mining constraints from examples labeled feasible and infeasible, and formulates it as mixed-integer linear programming (MILP) to produce constraints in symbolic, human-readable form \cite {pawlak2017automatic}.
Kolb et al. extend the inductive logic programming with the numerical variables to the CNF form, which is conjunctions of clauses over Boolean literals and linear inequalities, the problem learns the SMT(LRA) constraint from the feasible and infeasible set of the example and further proposed the algorithm INACL that can solve the problem by existing SMT solver \cite {kolb2018learning}.
Sched, Kolb, and Teso et al. propose INCALP, an algorithm to learn the hard constraints from data, by encoding constraint learning as a mixed-integer linear program. And if further extend from the work before by considering gradually larger subsets of examples, and terminating as soon as the suitable constraints are found \cite {schede2019learning}.
Meng et al. present the framework mining constraint from data via the given coefficients of the objective function and the corresponding solution, the proposed outer and inner properties successfully identify the constraint from the feasible set of example \cite {meng2021integer}.

Based on these studies, we use the concept of constraint learning and take the mixed integer linear programming as our constraint framework to develop the learn to bound strategy.

\subsection{Learning to Branch and Bound}
The recent works of learning strategy for branch and bound algorithms mostly focus on learning to branch. Balcan et al. aim to learn the branch strategy for branch and bound algorithms, they show how to use machine learning to determine an optimal weighting of any set of partitioning procedures for the instance distribution at a small number of samples \cite{balcan2018learning}. Zarpellon et al. propose a framework of Learn Branching Policies by introducing new input features and the imitation learning framework to generalize on the different problems in MILP architecture \cite{zarpellon2021parameterizing}. Qu et al. propose a reinforcement learning-based branch and bound algorithm, with the demonstration data collected by strong branch rule, the RL agent interacts with the solver with its learned policy that is significantly effective in performance improvement of branch and bound algorithm \cite{qu2022improved}. Rather than learning to branch, we use the learning to bound technique which is much more complicated since it might filter out feasible solutions, and successfully apply it to detecting k-plex.

\section{Problem Definition}
Our work combines constraint learning to learn the bound strategy for detecting k-plex. The problem definition includes detecting k-plex and mining bounding constraints from the search process.
\subsection{k-plex Detecting Problem}
Given a Graph \(G\) = \((V,E)\), a k-plex is set of vertices \(\{\)$v_{1}, v_{2}, v_{3},..., v_{m}$\(\}\) \( \subseteq S \), \((S\subseteq V)\), such that each of them has edges with all the others, and with the possible exception of up to k missing neighbors (including itself). If there exist a k-plex size is \(m\), each of the nodes in \(S\) has at least \(m\)-\(k\) neighbors in \(S\).

Finding maximum k-plex or enumerating all the maximal k-plex in a large network is time consuming. In this work, we define the k-plex detection problem as given the time threshold \(t\) and the lower bound \(lb\) of k-plex size, detecting the k-plex larger or equal than the lower bound within the time threshold. We evaluate the quality by judging if the algorithm can find out the k-plex that satisfied time and size constraints and how large the k-plex it can detect.
\subsection{Constraint Learning}
We first introduce the definition of constraint learning, then illustrate the Constraint Learning in Mixed Integer Programming (MILP) form and introduce the problem of finding constraints of bounding strategy in the branch and bound algorithm.

\subsubsection{Definition}
Given a space of possible instances \(X\), and instances are assignments to a set of variable \emph{Var} belonging to the unknown constraint theory; a space of possible constraints \(C\); an unknown target constraint theory \( T \subseteq  C\); a set of training instances \(E\), whereof positive instances satisfy \(T\), and negatives instances do not satisfy \(T\). Find a constraint theory \( H (H \subseteq  C)\) such that all instances in \(E\) are consistence with \(H\).  A constraint theory \(H\) is consistent with an example e if \(H(e) = label(e)\), here the constraint theory \(H\) can be seen as a function process on example \(e\) which classifies it to positive or negative. If the constraint theory \(H\) classifies the example e same as the original label of e, then \(H\) is consistent with an example \(e\).

\subsubsection{MILP Constraint Learning Problem} 
Recent works use the Mixed Integer Linear Programming framework as the constraint space. Let the example in examples set \(E\) is \(e\), the example set \(E\) consists of the positive instances set $e^{+}$ and the negative set  $e^{-}$. The example \(e\) = \(\{\)$x_{1}$, $x_{2}$, ..., $x_{n}$\(\}\), the $x_{1}$, $x_{2}$, ..., $x_{n}$ are process variables. The \(term\) $t_{j}$  is the function of process variables. Constraint $C_{i}$ is expressed in Equation \ref{algo:constraint_form}. The $w_{ij}$ $\in$ \(\mathbb{R}\) is weight of term $t_{j}$ of constraint $C_{i}$, and the $c_{i}$ $\in$ \(\mathbb{R}\) is a free term.
\begin{equation}\label{algo:constraint_form}
\sum_{j}{w_{ij}t_{j}\leq c_{i}}
\end{equation}
The MILP Constraint Learning Problem is built from the example set \(E\), and the constraint space \(C\). The solution to the problem is to find out a set of constraints \(C\) build from the weight $w_{ij}$ and free term $c_{i}$ such that the example set \(E\) consistence with it. Expressed in equation form, the constraints should satisfied the Equation \ref{algo:constraint_form positive} and Equation \ref{algo:constraint_form negative},where \(\epsilon\) is a small positive constant.

\begin{equation}\label{algo:constraint_form positive}
\sum_{j}{w_{ij}t_{j}(e)\leq c_{i}} \begin{matrix*} & \forall{e \in e^{+}},& \forall{C_{i} \in C} \end{matrix*}
\end{equation}
\begin{equation}\label{algo:constraint_form negative}
\sum_{j}{w_{ij}t_{j}(e) > c_{i}} + \epsilon \begin{matrix*} & \forall{e \in e^{-}},& \exists{C_{i} \in C} \end{matrix*}
\end{equation}

\subsection{Learning constraints of Bound Strategy}
Based on the above constraint learning definition, we propose a problem of learning bounding constraints in branch and bound algorithm for searching k-plex. Given the searching states in the branch and bound process, each state consists of the features of the intermediate solution set $V_{S}$ and the remaining set of candidate vertices $V_{A}$, and each state is labeled with stop examine (negative instances) or not keep examine (positive instances) according to the original branch and bound algorithm decide to explore it or not, we aim to find out a constraint theory \(H\) of with the specific variable \emph{Var} such that the given searching states (instances) are consist with \(H\). In 
Proposition \ref{prop:1}, we prove that we can always learned the constraints exactly distinguishes all the possible example. For applying to the detecting k-plex problem, it is not reasonable
to use lots of variables to find the exact constraint, our goal is to learn the bound strategy to be efficient. We expect the learned constraint has the power of the original bound strategy with a lower computation cost. This is done by design variables that have the certain power to express the bound instances without heavy computation cost, which leads to accelerating the branch and bound algorithm with a light constraint computation.

\begin{prop}
\label{prop:1}
Given the target constraint in the linear program formula $\sum_{j}{a_{ij}t_{j}(e)\leq b_{i}}$, it is always possible to find the leaned constraints such that correctly distinguish every example generate from target constraint under the the constraints space include the constraint space of target constraint.
\end{prop}

\begin{proof}
 By the definition that the linear program problem is infeasible if the feasible region is empty(there's no point satisfy all the constraints in the linear program problem) \cite{goemans1994linear}, \cite{mehlhorn2013linear}. Proof by contradiction, if there does not exist the solution correctly distinguishes every example generated from the target constraint, the linear program problem formed by the encoding of example generate from target constraints is infeasible, which means there does not exist a linear program formula $\sum_{j}{a_{ij}t_{j}(e)\leq b_{i}}$ satisfy the encoding of examples, which arrive to a contradiction.

\end{proof}
\section{Methodology}

Our approach aims to find a k-plex that satisfied the size constraints in a short time, to accomplish the goal, we want to let the computation of bound decision in the branch and bound process efficient and still retain the bound effect. We introduce the approach of learning bound strategy by constraint learning to accelerate the bound process. Note that our learn to bound mechanism can be applied to any bound strategy with the appropriate learning variables.

Our framework consists of the prepossessing of the input graph, the basic branch and bound algorithm, the learning process of learning constraints of bound strategy, and the overall learn to branch and bound algorithm.

\subsection{Preprocessing}
We adopt the pruning technique according to the \emph{Coreness} property and the \emph{Cliquness} property of k-plex from \cite{conte2017fast}. The \emph{Coreness} property indicated that for any k-plex size lager than \(m\), the neighbors of each vertex in the k-plex should larger than \(m\) - \(k\). Hence, before searching the k-plex larger than \(lb\) in the input Graph \(G\), we can discard the vertex whose degree is less than \(lb\) - \(k\) and also discard the edge connect to it. After the \emph{Coreness} pruning, we can move on to the \emph{Cliquness} property. The \emph{Cliquness} property indicated that for any vertex in the k-plex whose size is larger than \(m\), the vertex should exist in the clique whose size is larger or equal than \(\lceil m/k \rceil\), so we can discard the vertex not exist in a clique which size is larger or equal than \(\lceil lb/k \rceil\). According to the given value of \(lb\), we use these two pruning techniques to produce a new graph \(G^{'}\) as the preprocessing for the following search.

\subsection{Basic Branch and Bound Algorithm}
Here we use the simple branch and bound algorithm as our agent for collecting each search state for learning the bound strategy. 

We apply a simple heuristic for the branch strategy, here the $V_{S}$ and $V_{A}$ are the intermediate solution set and the remaining set of candidate vertices, respectively. As shown in Algorithm \ref{alg:4}, if the intermediate solution set  $V_{S}$ is null, we choose the vertex \(u\) to branch which the value of Equation \ref{algo:u} is max in the $V_{A}$; if the $V_{S}$ is not null, we choose the vertex \(u\) in $V_{A}$ which \(u\) has the most neighbor in the $V_{S}$ to branch.
\begin{equation}\label{algo:u}
\frac{1}{|u's \,neighbor|}\sum_{v}|v's \,neighboer| \begin{matrix*} & \forall{v \in u's neighbor}\end{matrix*}
\end{equation}

We adopt the idea of familiarity pruning strategy from \cite{yang2012socio}, this work is aim to find out the group with exact size \(p\), and each vertex on average can share no edge with at most \(k\) other vertices in query vertex \(F\). Therefore the familiarity pruning calculates the upper bound of every possible solution growing from $V_{S}$, if the upper bound indicated that on average each vertex cannot have the neighbor at least $p-k$, this state will be bound (Equation \ref{algo:familarity}). To adopt the equation calculating for the fixed p to our k-plex searching, we test the equation for \(p\) = \([ min \,(\,lb \,, \,|V_{S}|\, )\,,\,ub\, ]\) (\( ub \) is the max size of the solution can achieved), if for any \(p\) in this scope can not achieve a possible solution, the search of this state will be terminated. Here the \(N_{v}^{V_{S}}\) is the neighbors of v which the neighbor is in set \(V_{S}\); \(N_{v}^{V_{A}}\) is the neighbors of v which the neighbor is in set \(V_{A}\); InterEdge is the set of edges connect two vertexes in \(V_{S}\) and \(V_{A}\).

\begin{equation}\label{algo:familarity}
    \begin{aligned}
       \frac{1}{p} [\sum_{v \in V_{S}} |N_{v}^{V_{S}|} 
       + (p-|V_{S}|) 
       \max{{v \in V_{A}}}
       |N_{v}^{V_{A}}| + \\
       2\sum_{v \in V_{S}} 
       |InterEdge(v)|] 
       <(p-k-1) 
    \end{aligned}
\end{equation}

\begin{algorithm}[h]
	\caption{Basic Branch and Bound Algorithm}
	\label{alg:4}
        \begin{algorithmic}[] 
            \STATE \textbf{Input:} Graph G, k, lb
            \STATE \textbf{Output:} a k-plex solution
        \end{algorithmic}

	\begin{algorithmic}[1]
    	\STATE $G^{'} \gets$ Preprocessing$(G,k,lb)$ \
    	\STATE $V_{S}=\emptyset, V_{A}=G^{'}(V), Sol=\emptyset $ \
    	\STATE Basic Branch-and-Bound$(G^{'},V_{S},V_{A},k,lb)$ 
    	\RETURN largest solution in $Sol$
    	\newline
	\end{algorithmic}
	
    \begin{algorithmic}[] 
        \STATE Basic Branch-and-Bound $(G^{'},V_{S},V_{A},k,lb)$ 
    \end{algorithmic}   

	\begin{algorithmic}[1]
	    \IF{$|V_{S}| \geq lb $ and  $V_{S}$ is k-plex}
	        \STATE Add $V_{S}$ to solution set $Sol$
	    \ENDIF
	    
	    \WHILE {$V_{A} \neq \emptyset$}
    	    \IF{$V_{S} = \emptyset$}
    	        \STATE $u$ $\gets$ vertex has max value of Equation \ref{algo:u} in \(V_{A}\)
    	    \ELSE
    	        \STATE $u$ $\gets$ vertex has most neighbors in  \(V_{S}\)
    	    \ENDIF
    	    \STATE $V_{S}=V_{S} \cup u$
    	    \STATE $V_{A}=V_{A} \setminus u$
    	    \IF{PruningBasic(\(V_{S}\),\(V_{A}\)) is $true$}
    	        \STATE Record the pruning data:\\ \{Current state's features , False\} 
    		    \RETURN
    		\ELSE
    		    \STATE Record the not pruning data:\\ \{Current state's features , True\}
    		    \STATE Basic Branch-and-Bound$(G,V_{S},V_{A},k,lb)$
    		\ENDIF
	    \ENDWHILE
	    \RETURN
	\end{algorithmic}
\end{algorithm}

\subsection{Learning Constraints of Bound Strategy}
Constraint Learning learns the constraints from the given examples, we employ the idea of mining constraints from data to mining the bound strategy of the branch and bound process, hence we take each search state and its corresponding features as examples, and try to figure out the effective new bound strategy.
\subsubsection{Variable Chosen}
According to the equation in the familiarity pruning in the basic branch and bound algorithm, the calculation required the computation of the set of edges connecting any two vertices in $V_{S}$, the set of edges connecting any two vertices selected from $V_{A}$, and the set of edges connecting any two vertices in $V_{S}$  and the vertices selected from $V_{A}$. The edge computation is required to visit all the vertex and its neighbors both $V_{S}$ and $V_{A}$. We want to reduce these computations and retain the power of familiarity bound. We design a set of variables \emph{Var} to learn the constraints, described in Table 1, and these \emph{Var} generalize the original bounding strategy well and also reduce computation cost.

At the beginning of the search process, the size of $V_{A}$ tends to be large, to calculate the edges connecting any two vertices in $V_{A}$  is time consuming, we come up with using the variable of the average degree of the graph $G^{'}$ and the maximum degree of the vertex in the Graph $G^{'}$ to replace the $V_{A}$'s edges computation. These two variables only need to calculate one time in the whole algorithm, if we can construct the bounding strategy using these set of variables, the computation cost of each state can be significantly reduced.
\begin{table}
  \caption{Variable for Learning Bounding Constraints}
  \label{tab:var}
  \resizebox{\columnwidth}{!}{%
    \def\arraystretch{1.2}
  \begin{tabular}{cc}
    \toprule
    Variable & Description\\
    \midrule
    \texttt{lb}& lower bound \(lb\) of k-plex. \\ 
    \texttt{ub} &size of Graph $G^{'}$ after prepossessing. \\ 
    \texttt{k}&the k of k-plex \\ 
    \texttt{|$V_{S}$|} & size of the intermediate solution $V_{S}$.\\ 
    \texttt{$N_{V_{S}^{Max}}$}& max neighbors number of vertex in $V_{S}$. \\ 
    \texttt{$\sum$ $N_{V_{S}}$}& summation of vertex neighbors in $V_{S}$. \\ 
    \texttt{|$V_{A}$|}&size of candidate vertices in $V_{A}$ \\ 
    \texttt{|InterEdge|}& size of edges connect between $V_{S}$, $V_{A}$. \\ 
    \texttt{Avg. deg of $G^{'}$}& average degree of the vertices in $G^{'}$. \\ 
    \texttt{Max deg of $G^{'}$}&max degree of the vertex in $G^{'}$ .  \\
    \bottomrule
  \end{tabular}}
\end{table}

\subsubsection{MILP Framework}
Recall the MILP problem definition in Equation \ref{algo:constraint_form positive} and Equation \ref{algo:constraint_form negative}, the Figure \ref{fig:example} is the illustrate example of the MILP Constraint Learning. For all the example \(e\) $\in$ \(e^{+}\) should satisfy all the constraints in constraints model. While the example \(e\) $\in$ \(e^{-}\) should at least violate one constraint in the constraint model. In the examples set in Figure \ref{fig:example}, the two feasible examples satisfy all constraints in the constraint model; the first infeasible example violates the second constraint only, the second infeasible example violates the first and third constraints, the third infeasible example violates all three constraints. So, the constraint model in the figure is consistent with the examples set.

\begin{figure}[h]
\centering
\includegraphics[width=8.5cm]{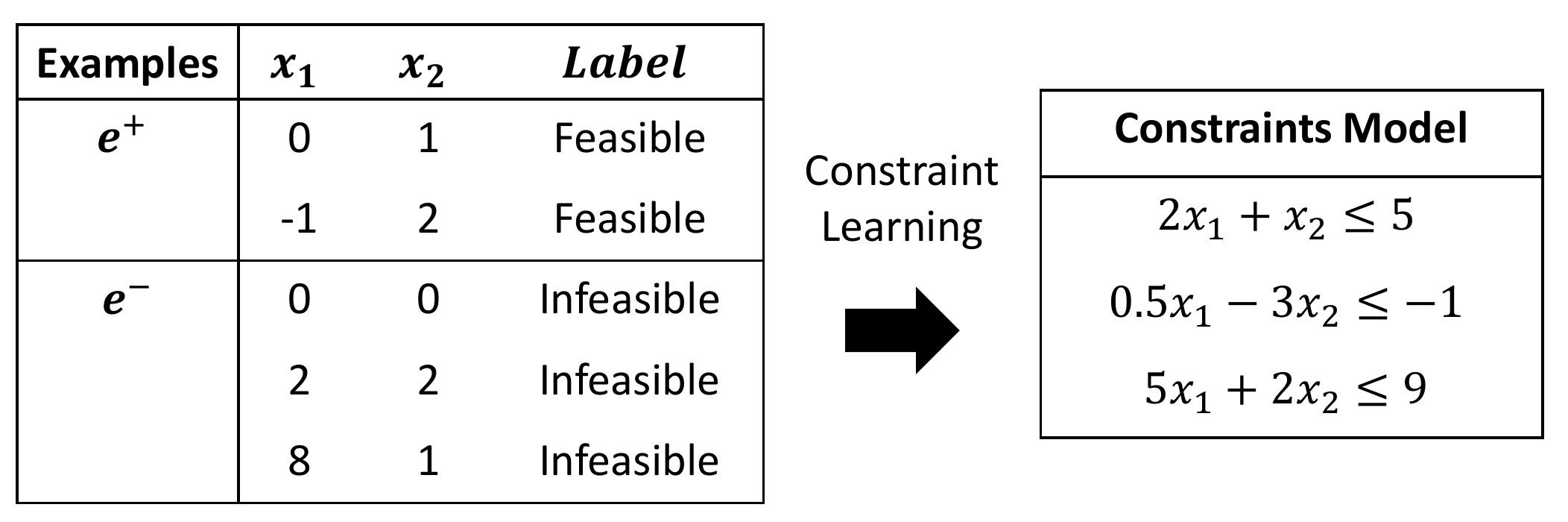}
\caption{Example of Constraint Learning}
\label{fig:example}
\end{figure}

For the application of learn the bound strategy, we set the constraint space has only one MILP formula, so the MILP constraint consistency equation can be simplify to Equation \ref{algo:constraint_form_sim positive} and Equation \ref{algo:constraint_form_sim negative}.
\begin{equation}\label{algo:constraint_form_sim positive}
\sum_{j}{w_{j}t_{j}(e)\leq c_{0} \begin{matrix*} & \forall{e \in e^{+}} \end{matrix*}}
\end{equation}
\begin{equation}\label{algo:constraint_form_sim negative}
\sum_{j}{w_{j}t_{j}(e) > c_{0}} + \epsilon \begin{matrix*} & \forall{e \in e^{-}} \end{matrix*}
\end{equation}
\(term\) \(t\) is the function of the process variable, here we set the process function \(term\) \(t\) is \(quadratic\) function, which means the \(t_{j}\) is build from the each variable itself and the product of the two variable in the set of variables. For the instances \(x=\{\)$x_{1},x_{2},...,x_{n}$ \(\}\), the \(t(x)\) is Equation \ref{algo:tx}
\begin{equation}\label{algo:tx}
    \begin{aligned}
t(x)= \{ x_{1},x_{2},...,x_{n},\\
x_{1}x_{1},x_{1}x_{2},...,x_{1}x_{n},\\
x_{2}x_{2},x_{2}x_{3},...,x_{2}x_{n},\\
...,\\
x_{n}x_{n} \} 
    \end{aligned}
\end{equation}
\subsubsection{Examples Extraction and Encoding}
Since the variable of our constraint learning is already decided, we compute and record the features (variable) of each search state, (the state is consist of current VS and VA), and record its label. If the familiarity pruning of basic branch and bound algorithm decide to keep searching from the state, the state is labeled with positive; otherwise, the state is labeled with negatives.

To learn the constraints from the data we extract from the basic branch and bound process, we need to encode these examples to a MILP formula so that the MILP solver can solve the problem. The encoding technique is adopted from \cite{pawlak2017automatic}, here we list the essential part of the encoding below, where \(l\) is an index of example. Equation \ref{algo:sli} indicated that the negative example should violate at least one constraint in constraint space, \(M\) is a large positive integer set to \(10^{6}\), $\epsilon$ is a small positive value set to \(10^{-6}\).
\begin{equation}\label{algo:constraint_form_encode positive}
\sum_{j}{w_{ij}t_{j}(e)\leq c_{i} \begin{matrix*} & \forall{e \in e^{+}}, &\forall{i} \end{matrix*}}
\end{equation}
\begin{equation}\label{algo:constraint_form_encode negative}
\sum_{j}{w_{ij}t_{j}(e) > M*S_{li} -M + c_{i}} + \epsilon \begin{matrix*} & \forall{e \in e^{-}}, &\forall{i} \end{matrix*}
\end{equation}
\begin{equation}\label{algo:sli}
\sum_{i}{S_{li} \geq 1 \begin{matrix*} & \forall{e \in e^{-}} \end{matrix*}} 
\end{equation}
\begin{equation}\label{algo:wij}
\begin{aligned}
w_{ij} \in [-1000,1000]\\
c_{i} \in [-1000,1000]\\
S_{li}\in \{0,1\}
\end{aligned}
\end{equation}

\subsection{Learn to Branch and Bound Algorithm}
Using the mechanism described above, our learn-to-bound algorithm is illustrated in Algorithm \ref{alg:learn}. First, we use the basic branch and bound algorithm to find k-plex on the random graph. During the searching process, we record each searching state with its corresponding feature and label. Then we use the MILP Constraint learning to find out the bounding strategy according to our design variable. While we obtain the constraint term from the constraint learning. We use this constraints term as our new decision strategy to decide whether keep searching from the current state.

\begin{algorithm}[]
\caption{Learn to Branch and Bound}
\label{alg:learn}
\begin{algorithmic}[] 
    \STATE \textbf{Input:} Graph G, k, lb
    \STATE \textbf{Output:} a k-plex solution
\end{algorithmic}

\begin{algorithmic}[1]
    \STATE Prepare \,other \,training \,graph \,\(G_{t}\) different \newline \,from \,the \,target \,Graph\ \(G\)
    \STATE $Data = \emptyset$ \
    \FOR{different lb, k}
        \STATE $G_{t}^{'} \gets$ Preprocessing$(G_{t},k,lb)$ \
        \STATE $V_{S}=\emptyset, V_{A}=G^{'}(V) $ \
        \STATE $Data$  $\gets$ $Data$ $\cup$ \newline Basic Branch-and-Bound $(G_{t}^{'},V_{S},V_{A},k,lb)$
    \ENDFOR
    \STATE $Bound Strategy$ $\gets$ \newline Learn Constraints(Data, Var, Constraint Space)
    \newline
    \STATE $G^{'} \gets$ Preprocessing$(G,k,lb)$ \
    \STATE $V_{S}=\emptyset, V_{A}=G^{'}(V), Sol=\emptyset $ \
    \STATE Learn-to-Bound$(G^{'},V_{S},V_{A},k,lb)$ 
    \newline
\end{algorithmic}
\begin{algorithmic}[]
   \STATE Learn-to-Bound $(G^{'},V_{S},V_{A},k,lb)$ 
\end{algorithmic}   

\begin{algorithmic}[1]
    \IF{$|V_{S}| \geq lb $ and  $V_{S}$ is k-plex}
        \STATE Add $V_{S}$ to solution set $Sol$
    \ENDIF
    
    \WHILE {$V_{A} \neq \emptyset$}
        \IF{$V_{S} = \emptyset$}
            \STATE $u$ $\gets$ vertex has max value of Equation \ref{algo:u} in \(V_{A}\)
        \ELSE
            \STATE $u$ $\gets$ vertex has most neighbors in  \(V_{S}\)
        \ENDIF
        \STATE $V_{S}=V_{S} \cup u$
        \STATE $V_{A}=V_{A} \setminus u$
        \IF{$BoundStrategy(State's Features)$ is $true$}
            \RETURN
        \ELSE
            \STATE Learn-to-Bound$(G,V_{S},V_{A},k,lb)$
        \ENDIF
    \ENDWHILE
    \RETURN
\end{algorithmic}
\end{algorithm}

\section{Experimental Results}
\label{sec:experiment}
 
\subsection{Experiment Setting}
\label{sec:exp_setup}
All algorithms are written in C++, except the constraint learning framework is written in python and solved by MILP solver Gurobi;  \textbf{S2V-DQN} written in python.  All the experiments are conducted on a desktop with an operating system in Intel(R) Core(TM) i9-10900F CPU @ 2.80GHz.

\subsubsection{Baseline}
For the comparing to our method, we use the recent baseline algorithm \textbf{EnumFast} \cite{conte2017fast}, \textbf{Enum} \cite{berlowitz2015efficient},  \textbf{Maplex} \cite{zhou2021improving}, \textbf{S2V-DQN} \cite{khalil2017learning}, \textbf{Basic BnB}: our basic branch and bound algorithm, compare to \textbf{Learn BnB}: our learn to bound algorithm.

The Enum algorithm first extends the subset from the first input node by scanning all the nodes once and adding the vertices if the subset remains k-plex. Then it extends the possible combination from the first k-plex by reducing the problem of enumerating k-plex for an arbitrary graph G to enumerating k-plex for the graph that almost satisfies k-plex. The FastEnum adopted the same strategy as Enum but propose two preprocessing techniques to prune the infeasible vertices which make the searching more efficient. The Maplex is the branch and bound based algorithm, it proposes the graph color bound to and other reduction techniques to detecting k-plex. We also take a representative reinforcement learning method S2V-DQN framework and modify its termination condition as if the subset still a k-plex or not for detecting k-plex.

\subsubsection{Datasets}
We conduct the experiment on \emph{DBLP}, \emph{web-arabic-2005} and four collaboration networks: \emph{ca-Erdos992}, \emph{ca-GrQc}, \emph{ca-Condmat} and \emph{ca-AstroPh} to test the algorithm performance with the different graph properties. The properties of these four graphs are shown in table \ref{tab:ca}. We also test the Basic BnB and  Learn BnB on several small networks \cite{nr}.

\begin{table}
  \caption{Properties of the collaboration networks.}
  \label{tab:ca}
  \begin{tabular}{ccccc}
    \toprule
    Dataset &$|V|$ & $|E|$ &avg. degree &max degree\\
    \midrule
    ca-Erdos992 & 6100 & 7515 & 2 &61 \\
    ca-GrQc & 5242 & 28980& 5& 81 \\
    ca-Condmat & 23133 & 186936 & 8 & 279 \\
    ca-AstroPh & 18772& 396160 & 21 & 504 \\
    \bottomrule
  \end{tabular}
\end{table}
\subsubsection{Constraint Learning Detail} 
Here we illustrate the detail of the learning constraint algorithm. Since learning the bound strategy can process offline. We only need to learn the bound strategy ones and we can use the constraints we learn to perform the branch and bound algorithm on all the datasets.

To make the learning constraints can be generalized on different graph properties, it would be better to prepare the multiple graph data with different parameters for the constraint learning. But considering our first intention of accelerating the search process of detecting k-plex, we wonder if we can only use the light datasets and the small number of parameters to build the constraints. Under this intention we only use the random graph whose size is \{100, 150, 200, 250\}, and each size has two graphs, and the k is set to \{2, 4\}, the lower bound \(lb\) is set to 5 only. For each graph with the corresponding k, we only run the basic branch and bound algorithm in 60 seconds, and collect the example during these training processes. After the example extraction, we encode the examples to MILP format to formulate a MILP problem, then use the solver Gurobi to solve the problem in 300 seconds. The solution that comes up by the MILP solver is our learn to bound strategy. The total time for the learning constraints of the bound strategy is in 21 minutes. Note that we only need to learn the bound strategy once, then it can apply to the experiment of different networks.

\subsubsection{Setting} 
For all the baseline algorithms (expect the \textbf{Enum}) and our \textbf{Basic BnB}, \textbf{Learn BnB}, we first conduct the preprocessing on the original graph according to the input \(lb\); for the \textbf{Maplex}, we further let it conduct its preprocessing algorithm after it, and we do not count this extra preprocessing time.

For the experiment on the collaboration networks, since the network size is small, we only test the timeout in 10 seconds and 60 seconds. Corresponding to the size of the k-plex each graph can detect, we set the size lower bound \(lb\) in [5, 10, 20, 30, 40] and test on k= \{2, 4\}.
For the experiment on the \emph{DBLP} and \emph{web-arbric-2005}, we set the k in \{2, 4, 8, 16, 32 \}, to see how the different values of k influence the algorithm performance; the lower bound \(lb\) is set to \{30, 50, 75, 100 \} to evaluate if the algorithm can detect k-plex within the time limit. The time limit is set to 10 seconds,60 seconds, and 300 seconds.

\subsection{Performance Evaluation}

First, we conduct the experiment on \emph{DBLP} network, \emph{DBLP} has 317,080 vertices, 1,049,866 edges. the result is in Table \ref{tab:dblp}. The results show our Learn BnB algorithm can detect the k-plex satisfy the lower bound constraints in the time limit with any hyper parameter, while the other baseline can not produce the k-plex in some cases. The Enum algorithm can not produce any solution since the algorithm without preprocessing is hard to perform on large networks. The FastEnum algorithm only has solutions in k=2 and lb=75, 100. It seems that the enumeration-based algorithm struggles to generate the combination of possible solutions when the value of k grows. The Maplex, Basic BnB and Learn BnB are branch and bound based algorithms. These methods seem to work better than the enumeration algorithm in this network. But the Maplex and Basic BnB performance also decrease when k is large and lb is small. Our Learn BnB still extracts the k-plex in this situation. We thought the value of \(lb\) and \(k\) have a different impact on the baseline algorithm. The S2V-DQN work well under some parameters, but it has the obstacle that we can not load the whole graph into memory, so the RL method only can work on the case that the graph after preprocessing are small enough it also implies RL bases method are resource consuming and our Learn BnB are an more generalized method for solving k-plex. 
\begin{table*}
  \caption{The experimental results of DBLP. The value is the maximize k-plex size found by each methods under different value of \emph{k}, \emph{lb } and \emph{t}. Larger values represent better performance.}\label{tab:dblp}
  \begin{tabular}{cccccccccccccc}
    \toprule
    \multicolumn{1}{c}{\emph{method}}  &\multicolumn{12}{c}{\emph{k-plex size}} \\
    \midrule
    &\emph{lb} &  \multicolumn{3}{c}{30} & \multicolumn{3}{c}{50} & \multicolumn{3}{c}{75} & \multicolumn{3}{c}{100}\\
    \cmidrule{2-14}
     &\emph{time limit} & \multicolumn{1}{c}{10(s)} & \multicolumn{1}{c}{60(s)} & \multicolumn{1}{c}{300(s)}  & \multicolumn{1}{c}{10(s)} & \multicolumn{1}{c}{60(s)} & \multicolumn{1}{c}{300(s)} & \multicolumn{1}{c}{10(s)} & \multicolumn{1}{c}{60(s)} & \multicolumn{1}{c}{300(s)} & \multicolumn{1}{c}{10(s)} & \multicolumn{1}{c}{60(s)} & \multicolumn{1}{c}{300(s)} \\
    \cmidrule{2-14}
    & \emph{k} & \\
    \cmidrule{2-2}
    FastEnum (KDD'17) & & 43 & 43 & 43  & 0 & 0 & 0 & \textbf{114} & \textbf{114} & \textbf{114}  & \textbf{114} & \textbf{114} & \textbf{114} \\
    Enum (SIGMOD'15)& & 0 & 0 & 0 & 0 & 0 & 0 & 0 & 0 & 0 & 0 & 0 & 0 \\
    Maplex (AAAI'21)& 2&  0 & 0 & 62  & 0 & 0 & 84 & \textbf{114} & \textbf{114} & \textbf{114}  & \textbf{114} & \textbf{114} & \textbf{114}\\
    S2V-DQN (NIPS'17) && 29 & \textbf{114} & \textbf{114}  & \textbf{114}& \textbf{114} & \textbf{114} & \textbf{114} & \textbf{114} & \textbf{114} &\textbf{114} & \textbf{114} & \textbf{114}\\
    Base BnB (\textbf{Ours}) & & 0 & 0 & \textbf{114}  & 0 & \textbf{114} & \textbf{114} & \textbf{114} & \textbf{114} & \textbf{114}  & \textbf{114} & \textbf{114} & \textbf{114}\\
    Learn BnB (\textbf{Ours}) & &  \textbf{114} & \textbf{114} & \textbf{114}  & \textbf{114} & \textbf{114} & \textbf{114} & \textbf{114} & \textbf{114} & \textbf{114}  & \textbf{114} & \textbf{114} & \textbf{114}\\
    \cmidrule{1-14}
    FastEnum (KDD'17) & & 0 & 0 & 0  & 0 & 0 & 0 & 0 & 0 & 0  & 0 & 0 & 0 \\ 
    Enum (SIGMOD'15) & &  0 & 0 & 0  & 0 & 0 & 0 & 0 & 0 & 0  & 0 & 0 & 0 \\
    Maplex (AAAI'21) & 4&  0 & 0 & 39  & 0 & 0 & 84 & 0 & 0 & \textbf{114}  & 0 & 0 & \textbf{114} \\
    S2V-DQN (NIPS'17)& &  22 & \textbf{114} & \textbf{114}  & \textbf{114}& \textbf{114} & \textbf{114} & \textbf{114} & \textbf{114} & \textbf{114} &\textbf{114} & \textbf{114} & \textbf{114} \\
    Base BnB (\textbf{Ours})& &  0 & 0 & 94  & 0 & \textbf{114} & \textbf{114} & \textbf{114} & \textbf{114} & \textbf{114}  & \textbf{114} & \textbf{114} & \textbf{114} \\
    Learn BnB (\textbf{Ours}) & &  \textbf{78} & \textbf{114} & \textbf{114}  & \textbf{114} & \textbf{114} & \textbf{114} & \textbf{114} & \textbf{114} & \textbf{114}  & \textbf{114} & \textbf{114} & \textbf{114} \\
    \cmidrule{1-14}
    FastEnum (KDD'17) & & 0 & 0 & 0  & 0 & 0 & 0 & 0 & 0 & 0  & 0 & 0 & 0 \\
    Enum (SIGMOD'15) & &  0 & 0 & 0  & 0 & 0 & 0 & 0 & 0 & 0  & 0 & 0 & 0 \\
    Maplex (AAAI'21) &8 &  0 & 0 & 0  & 0 & 0 & 84 & 0 & 0 & \textbf{114}  & 0 & 0 & \textbf{114} \\
    S2V-DQN (NIPS'17)& &  13 & \textbf{114} & \textbf{114}  & \textbf{114}& \textbf{114} & \textbf{114} & \textbf{114} & \textbf{114} & \textbf{114} &\textbf{114} & \textbf{114} & \textbf{114} \\
    Base BnB (\textbf{Ours})& &  0 & 0 & 50  & 0 & \textbf{114} & \textbf{114} & \textbf{114} & \textbf{114} & \textbf{114}  & \textbf{114} & \textbf{114} & \textbf{114} \\
    Learn BnB (\textbf{Ours})& &  \textbf{73} & \textbf{114} & \textbf{114}  & \textbf{114} & \textbf{114} & \textbf{114} & \textbf{114} & \textbf{114} & \textbf{114}  & \textbf{114} & \textbf{114} & \textbf{114} \\
    \cmidrule{1-14}
    FastEnum (KDD'17) & & 0 & 0 & 0  & 0 & 0 & 0 & 0 & 0 & 0  & 0 & 0 & 0 \\
    Enum (SIGMOD'15) & &  0 & 0 & 0  & 0 & 0 & 0 & 0 & 0 & 0  & 0 & 0 & 0\\
    Maplex (AAAI'21) & 16 &  0 & 0 & 0  & 0 & 0 & 0 & 0 & 0 & 0  & 0 & 0 & \textbf{114} \\
    S2V-DQN (NIPS'17)& &  - & - & -  & \textbf{114}& \textbf{114} & \textbf{114} & \textbf{114} & \textbf{114} & \textbf{114} &\textbf{114} & \textbf{114} & \textbf{114} \\
    Base BnB (\textbf{Ours})& &  0 & 0 & 0  & 0 & 67 & \textbf{114} & \textbf{114} & \textbf{114} & \textbf{114}  & \textbf{114} & \textbf{114} & \textbf{114} \\
    Learn BnB (\textbf{Ours})& &  \textbf{112} & \textbf{114} & \textbf{114} & \textbf{114} & \textbf{114} & \textbf{114} & \textbf{114} & \textbf{114} & \textbf{114}  & \textbf{114} & \textbf{114} & \textbf{114} \\
    \cmidrule{1-14}
    FastEnum (KDD'17) & & - &-  &-   & 0 & 0 & 0 & 0 & 0 & 0  & 0 & 0 & 0 \\
    Enum (SIGMOD'15) & & - & - & -  & 0 & 0 & 0 & 0 & 0 & 0  & 0 & 0 & 0 \\
    Maplex (AAAI'21) &32 & -  & - & -  & 0 & 0 & 0 & 0 & 0 & 0  & 0 & 0 & \textbf{114} \\
    S2V-DQN (NIPS'17)& &  - & - & -  & -&- & - & \textbf{114} & \textbf{114} & \textbf{114} &\textbf{114} & \textbf{114} & \textbf{114}\\
    Base BnB (\textbf{Ours})& & -  &-  &-   & 0 & 0 & 0 & 0 & \textbf{114} & \textbf{114}  & \textbf{114} & \textbf{114} & \textbf{114} \\
    Learn BnB (\textbf{Ours})& & -  &-  & -  & \textbf{63} & \textbf{114} & \textbf{114} & \textbf{114} & \textbf{114} & \textbf{114}  & \textbf{114} & \textbf{114} & \textbf{114} \\

    \bottomrule
  \end{tabular}
  \begin{tablenotes}\footnotesize
    \item[*] *Since when k=32, finding k-plex of lb=30 is trivial, we do not conduct experiments under these parameters.
    \item[*] *For the method \emph{S2V-DQN}, we can not load the whole graph into memory if the graph after prerocessing is to large, so there's no result under certain lb and t.
    \end{tablenotes}
\end{table*}
Next, we conduct the experiment on \emph{web-arabic-2005} networks for confirming our thoughts. The \emph{web-arabic-2005} has 163598 vertices and 1747269 edges. The result is in Table \ref{tab:web}. The results seem similar to the result of \emph{DBLP} that our Learn BnB outperforms other baselines in most of the cases. But there's a difference in the enumeration method FastEnum when the \(lb\) is small. The first step of the FastEnum is to extend the subset from the first vertex and scan the whole vertices, if adding the current visit vertex can make the subset still retain a k-plex, then it keeps adding the node into the subset. This naive technique is relatively easy to detect k-plex in a short time if the density of graphs is high, and it also depends on the input vertices order. Although FastEnum detects the k-plex in 10 seconds, it is hard to further extend the solution even execute to 300 seconds. In contrast, our Learn BnB did not detect the k-plex in 10 seconds but was able to expand the k-plex to a large size with the time limit growing. The S2V-DQN can not work in more case in \emph{web-arabic-2005} then \emph{DBLP}, since the \emph{web-arabic-2005} has high graph density and thus the graph after preprocessing are larger than \emph{DBLP}.
\begin{table*}
  \caption{The experimental results of web-arabic-2005. The value is the maximize k-plex size found by each methods under different value of \emph{k}, \emph{lb } and \emph{t}. Larger values represent better performance.}
\label{tab:web}
  \begin{tabular}{cccccccccccccc}
    \toprule
    \multicolumn{1}{c}{\emph{method}}  &\multicolumn{12}{c}{\emph{k-plex size}} \\
    \midrule
    &\emph{lb} &  \multicolumn{3}{c}{30} & \multicolumn{3}{c}{50} & \multicolumn{3}{c}{75} & \multicolumn{3}{c}{100}\\
    \cmidrule{2-14}
     &\emph{time limit} & \multicolumn{1}{c}{10(s)} & \multicolumn{1}{c}{60(s)} & \multicolumn{1}{c}{300(s)}  & \multicolumn{1}{c}{10(s)} & \multicolumn{1}{c}{60(s)} & \multicolumn{1}{c}{300(s)} & \multicolumn{1}{c}{10(s)} & \multicolumn{1}{c}{60(s)} & \multicolumn{1}{c}{300(s)} & \multicolumn{1}{c}{10(s)} & \multicolumn{1}{c}{60(s)} & \multicolumn{1}{c}{300(s)} \\
    \cmidrule{2-14}
    & \emph{k} & \\
    \cmidrule{2-2}
    FastEnum & & \textbf{39} & 39 & 39  & 0 & 0 & 50 & 88 & 88 & 88  & \textbf{102} & \textbf{102} & \textbf{102} \\
    Enum & &  0 & 0 & 0  & 0 & 0 & 0 & 0 & 0 & 0  & 0 & 0 & 0 \\
    Maplex & 2&  0 & 0 & 0  & 0 & 0 & 0 & 0 & 0 & 0  & 0 & 0 & \textbf{102} \\
    S2V-DQN & &  - & - & -  & -& - & - & \textbf{102} & \textbf{102} & \textbf{102} &\textbf{102} & \textbf{102} & \textbf{102} \\
    Base BnB & &  0 & 0 & 0  & 0 & 0 & 0 & 0 & \textbf{102} & \textbf{102}  & \textbf{102} & \textbf{102} & \textbf{102} \\
    Learn BnB & &  0 & \textbf{45} & \textbf{102}  & 0 & \textbf{64} & \textbf{102} & \textbf{102} & \textbf{102} & \textbf{102} & \textbf{102} & \textbf{102} & \textbf{102} \\
    \cmidrule{1-14}
    FastEnum & & \textbf{39} & 39 & 39  & 0 & 0 & 0 & 88 & 88 & 88  & \textbf{102} & \textbf{102} & \textbf{102} \\
    Enum & &  0 & 0 & 0  & 0 & 0 & 0 & 0 & 0 & 0  & 0 & 0 & 0 \\
    Maplex & 4&  0 & 0 & 0  & 0 & 0 & 0 & 0 & 0 & \textbf{102} & 0 & \textbf{102} & \textbf{102} \\
    S2V-DQN & &  - & - & -  & -& - & - & \textbf{102} & \textbf{102} & \textbf{102} &\textbf{102} & \textbf{102} & \textbf{102} \\
    Base BnB & &  0 & 0 & 0  & 0 & 0 & 0 & 0 & \textbf{102} & \textbf{102} & \textbf{102} & \textbf{102} & \textbf{102} \\
    Learn BnB & &  0 & \textbf{44} & \textbf{101}  & 0 & \textbf{55} & \textbf{102} & \textbf{102} & \textbf{102} & \textbf{102}  & \textbf{102} & \textbf{102} & \textbf{102} \\
    \cmidrule{1-14}
    FastEnum & & \textbf{39} & 39 & 39  & 0 & 0 & 0 & 88 & 88 & 88  & \textbf{102}& \textbf{102} & \textbf{102} \\
    Enum & &  0 & 0 & 0  & 0 & 0 & 0 & 0 & 0 & 0  & 0 & 0 & 0 \\
    Maplex &8 &  0 & 0 & 0  & 0 & 0 & 0 & 0 & 0 & \textbf{102} & 0 & 0 & \textbf{102} \\
    S2V-DQN & &  - & - & -  & -& - & - & \textbf{102} & \textbf{102} & \textbf{102} &\textbf{102} & \textbf{102} & \textbf{102} \\
    Base BnB & &  0 & 0 & 0  & 0 & 0 & 0 & 0 & \textbf{102} & \textbf{102} & \textbf{102} & \textbf{102} & \textbf{102} \\
    Learn BnB & &  0 & \textbf{43} & \textbf{98}  & 0 & \textbf{51} &\textbf{102} & \textbf{102} & \textbf{102} & \textbf{102}  & \textbf{102} & \textbf{102} & \textbf{102} \\
    \cmidrule{1-14}
    FastEnum & & \textbf{37} & 37 & 37  & 0 & 0 & 0 & 88 & 88 & 88  & 0 & 0 & 0 \\
    Enum & &  0 & 0 & 0  & 0 & 0 & 0 & 0 & 0 & 0  & 0 & 0 & 0 \\
    Maplex & 16&  0 & 0 & 0  & 0 & 0 & 0 & 0 & \textbf{102} & \textbf{102}& 0 & \textbf{102} & \textbf{102} \\
    S2V-DQN & &  - & - & -  & -& - & -& \textbf{102} & \textbf{102} & \textbf{102} &\textbf{102} & \textbf{102} & \textbf{102} \\
    Base BnB & &  0 & 0 & 0  & 0 & 0 & 0 & 0 & \textbf{102} & \textbf{102}  & \textbf{102} & \textbf{102} & \textbf{102} \\
    Learn BnB & &  0 & \textbf{41} & \textbf{94}  & \textbf{102} & \textbf{102} &\textbf{102}& \textbf{102} & \textbf{102} &\textbf{102} & \textbf{102}& \textbf{102} & \textbf{102} \\
    \cmidrule{1-14}
    FastEnum &  & - & - & -  & 0 & 0 & 0 & 0 & 0 & 0  & 0 & 0 & 0 \\
    Enum & & -  & - & -  & 0 & 0 & 0 & 0 & 0 & 0  & 0 & 0 & 0 \\
    Maplex & 32& -  & - &  - & 0 & 0 & 0 & 0 & 0 & 0  & 0 & 0 & \textbf{102} \\
    S2V-DQN & &  - & - & -  & -& - & - & - & - & - &\textbf{102} & \textbf{102} & \textbf{102} \\
    Base BnB & &  - &-  & -  & 0 & 0 & 0 & 0 & 0 & 0  & 0 & \textbf{102} & \textbf{102} \\
    Learn BnB & & -  &-  &  - & 0 & 0 & \textbf{96} & 0 & 0 &\textbf{102}  & \textbf{102} & \textbf{102} & \textbf{102} \\
    \bottomrule
  \end{tabular}
  \begin{tablenotes}\footnotesize
    \item[*] *Since when k=32, finding k-plex of lb=30 is trivial, we do not conduct experiments under these parameters.
    \item[*] *For the method \emph{S2V-DQN}, we can not load the whole graph into memory if the graph after prerocessing is to large, so there's no result under certain lb and t.
    \end{tablenotes}
\end{table*}
Although our goal is to extract k-plex in a short time, we further conduct the experiment on \emph{DBLP} and \emph{web-arabic-2005} in a long time to see can each algorithm perform better in the large time scope. We test each algorithm in 10,20,30 minute and show in Table \ref{tab:dblp_long} and Table \ref{tab:web_long}. We can see all algorithm has some improvements with giving a long time. But none of the baselines outperforms our approach and thus the learn to bound efficiently and correctly extract the k-plex even in large time scope.
\begin{table*}
  \caption{The experimental results of DBLP under long time limit. The value is the maximize k-plex size found by each methods under different value of \emph{k}, \emph{lb } and \emph{t}. Larger values represent better performance.}\label{tab:dblp_long}
  \centering
  \resizebox{18cm}{!}{%
\def\arraystretch{1.2}
  \begin{tabular}{cccccccccccccc}
    \toprule
    \multicolumn{1}{c}{\emph{method}}  &\multicolumn{12}{c}{\emph{k-plex size}} \\
    \midrule
    &\emph{lb} &  \multicolumn{3}{c}{30} & \multicolumn{3}{c}{50} & \multicolumn{3}{c}{75} & \multicolumn{3}{c}{100}\\
    \cmidrule{2-14}
     &\emph{time limit} & \multicolumn{1}{c}{10(min)} & \multicolumn{1}{c}{20(min)} & \multicolumn{1}{c}{30(min)}  & \multicolumn{1}{c}{10(min)} & \multicolumn{1}{c}{20(min)} & \multicolumn{1}{c}{30(min)} & \multicolumn{1}{c}{10(min)} & \multicolumn{1}{c}{20(min)} & \multicolumn{1}{c}{30(min)} & \multicolumn{1}{c}{10(min)} & \multicolumn{1}{c}{20(min)} & \multicolumn{1}{c}{30(min)}  \\
    \cmidrule{2-14}
    & \emph{k} & \\
    \cmidrule{2-2}
    FastEnum & & 43 & 43 & 43  & 0 & 0 & 0& \textbf{114} & \textbf{114} & \textbf{114}  & \textbf{114} & \textbf{114} & \textbf{114} \\
    Enum & &  0 & 0 & 0 & 0 & 0 & 0 & 0 & 0 & 0  & 0 & 0 & 0 \\
    Maplex &2 &  64 & 80 & 80  & 84 & 84 & 84 & \textbf{114} & \textbf{114} & \textbf{114}  & \textbf{114} & \textbf{114} & \textbf{114} \\
    Base BnB (\textbf{Ours}) & &  \textbf{114} & \textbf{114} & \textbf{114}  & \textbf{114} & \textbf{114} & \textbf{114} & \textbf{114} & \textbf{114} & \textbf{114}  & \textbf{114} & \textbf{114} & \textbf{114} \\
    Learn BnB (\textbf{Ours}) & &  \textbf{114} & \textbf{114} & \textbf{114}  & \textbf{114} & \textbf{114} & \textbf{114} & \textbf{114} & \textbf{114} & \textbf{114}  & \textbf{114} & \textbf{114} & \textbf{114} \\
    \cmidrule{1-14}
    FastEnum  & & 0 & 0 & 0  & 0 & 0 & 0 & 0 & 0 & 0  & 0 & 0 & 0 \\
    Enum & &  0 & 0 & 0  & 0 & 0 & 0 & 0 & 0 & 0  & 0 & 0 & 0 \\
    Maplex & 4&  0 & 80 & 80  & 84 & 84 & 84 & \textbf{114} & \textbf{114} & \textbf{114}  & \textbf{114} & \textbf{114} & \textbf{114} \\
    Base BnB (\textbf{Ours})& &  0 & 102 & \textbf{114}  & \textbf{114} & \textbf{114} & \textbf{114} & \textbf{114} & \textbf{114} & \textbf{114}  & \textbf{114} & \textbf{114} & \textbf{114} \\
    Learn BnB (\textbf{Ours}) & &   \textbf{114} & \textbf{114} & \textbf{114}  & \textbf{114} & \textbf{114} & \textbf{114} & \textbf{114} & \textbf{114} & \textbf{114}  & \textbf{114} & \textbf{114} & \textbf{114} \\
    \cmidrule{1-14}
    FastEnum & & 0 & 0 & 0  & 0 & 0 & 0 & 0 & 0 & 0  & 0 & 0 & 0 \\ 
    Enum & &  0 & 0 & 0  & 0 & 0 & 0 & 0 & 0 & 0  & 0 & 0 & 0 \\
    Maplex  & 8&  0 & 63 & 63  & 84 & 84 & 84 & \textbf{114} & \textbf{114} & \textbf{114}  & \textbf{114} & \textbf{114} & \textbf{114} \\
    Base BnB (\textbf{Ours})& &  103 & \textbf{114} & \textbf{114}  & \textbf{114} & \textbf{114} & \textbf{114} & \textbf{114} & \textbf{114} & \textbf{114}  & \textbf{114} & \textbf{114} & \textbf{114} \\
    Learn BnB (\textbf{Ours})& &  \textbf{114} & \textbf{114} & \textbf{114}  & \textbf{114} & \textbf{114} & \textbf{114} & \textbf{114} & \textbf{114} & \textbf{114}  & \textbf{114} & \textbf{114} & \textbf{114} \\
    \cmidrule{1-14}
    FastEnum  & &0 & 0 & 0  & 0 & 0 & 0 & 0 & 0 & 0  & 0 & 0 & 0 \\
    Enum  & &  0 & 0 & 0  & 0 & 0 & 0 & 0 & 0 & 0  & 0 & 0 & 0 \\
    Maplex &16 &  63 & 63 & 63  & 0 & 0 & 62 & 0 & 0 & 0  & \textbf{114} & \textbf{114} & \textbf{114} \\
    Base BnB (\textbf{Ours})& &  0 & 0 & 47  & \textbf{114} & \textbf{114} & \textbf{114} & \textbf{114} & \textbf{114} & \textbf{114}  & \textbf{114} & \textbf{114} & \textbf{114} \\
    Learn BnB (\textbf{Ours})& &  \textbf{114} & \textbf{114} & \textbf{114} & \textbf{114} & \textbf{114} & \textbf{114} & \textbf{114} & \textbf{114} & \textbf{114}  & \textbf{114} & \textbf{114} & \textbf{114} \\
    \cmidrule{1-14}
    FastEnum  &  & - &-  &-   & 0 & 0 & 0 & 0 & 0 & 0  & 0 & 0 & 0 \\
    Enum  & & - & - & -  & 0 & 0 & 0 & 0 & 0 & 0  & 0 & 0 & 0 \\
    MaPlex  &32 & -  & - & -  & 0 & 0 & 0 & 0 & 0 & 0  & \textbf{114} & \textbf{114} & \textbf{114} \\
    Base BnB (\textbf{Ours})& & -  &-  &-   & 0 & 100 & \textbf{114} & \textbf{114} & \textbf{114} & \textbf{114}  & \textbf{114} & \textbf{114} & \textbf{114} \\
    Learn BnB (\textbf{Ours})& & -  &-  & -  & \textbf{114} & \textbf{114} & \textbf{114} & \textbf{114} & \textbf{114} & \textbf{114}  & \textbf{114} & \textbf{114} & \textbf{114} \\
    \bottomrule
  \end{tabular}}
  \begin{tablenotes}\footnotesize
    \item[*] *Since when k=32, finding k-plex of lb=30 is trivial, we do not conduct experiments under these parameters.
    \item[*] *For the method \emph{S2V-DQN}, we can not load the whole graph into memory if the graph after prerocessing is to large, so there's no result under certain lb and t.
    \end{tablenotes}
\end{table*}

\begin{table*}
  \caption{The experimental results of web-arabic-2005 under long time limit. The value is the maximize k-plex size found by each methods under different value of \emph{k}, \emph{lb } and \emph{t}. Larger values represent better performance.}
\label{tab:web_long}
  \centering
  \resizebox{18cm}{!}{%
\def\arraystretch{1.2}
  \begin{tabular}{cccccccccccccc}
    \toprule
    \multicolumn{1}{c}{\emph{method}}  &\multicolumn{12}{c}{\emph{k-plex size}} \\
    \midrule
    &\emph{lb} &  \multicolumn{3}{c}{30} & \multicolumn{3}{c}{50} & \multicolumn{3}{c}{75} & \multicolumn{3}{c}{100}\\
    \cmidrule{2-14}
     &\emph{time limit} & \multicolumn{1}{c}{10(min)} & \multicolumn{1}{c}{20(min)} & \multicolumn{1}{c}{30(min)}  & \multicolumn{1}{c}{10(min)} & \multicolumn{1}{c}{20(min)} & \multicolumn{1}{c}{30(min)} & \multicolumn{1}{c}{10(min)} & \multicolumn{1}{c}{20(min)} & \multicolumn{1}{c}{30(min)} & \multicolumn{1}{c}{10(min)} & \multicolumn{1}{c}{20(min)} & \multicolumn{1}{c}{30(min)}  \\
    \cmidrule{2-14}
    & \emph{k} & \\
    \cmidrule{2-2}
    FastEnum & & 39 & 39 & 39  & 0 & 0 & 0 & 88 & 88 & 88 & \textbf{102} & \textbf{102} & \textbf{102} \\ 
    Enum & &  0 & 0 & 0  & 0 & 0 & 0 & 0 & 0 & 0  & 0 & 0 & 0 \\
    Maplex & 2&  0 & 0 & 0 & 0 & \textbf{102} & \textbf{102} & \textbf{102} &\textbf{102} & \textbf{102}  & \textbf{102} & \textbf{102} & \textbf{102} \\
    Base BnB & &  0 & 0 & 0  & 0 & 0 & 0 & \textbf{102} & \textbf{102} & \textbf{102}  & \textbf{102} & \textbf{102} & \textbf{102} \\
    Learn BnB & &  \textbf{102} &\textbf{102} & \textbf{102}  & \textbf{102} & \textbf{102} & \textbf{102} & \textbf{102} & \textbf{102} & \textbf{102} & \textbf{102} & \textbf{102} & \textbf{102} \\
    \cmidrule{1-14}
    FastEnum &  & 39 & 39 & 39  & 0 & 0 & 0 & 88 & 88 & 88  & \textbf{102} & \textbf{102} & \textbf{102} \\
    Enum & &  0 & 0 & 0  & 0 & 0 & 0 & 0 & 0 & 0  & 0 & 0 & 0 \\
    Maplex &4 &  0 & 0 & 0  & 0 & 0 & 0 &\textbf{102}& \textbf{102} & \textbf{102} & \textbf{102} & \textbf{102} & \textbf{102} \\
    Base BnB & &  0 & 0 & 0  & 0 & 0 & 0 & \textbf{102} & \textbf{102} & \textbf{102} & \textbf{102} & \textbf{102} & \textbf{102} \\
    Learn BnB & &  \textbf{102} & \textbf{102} & \textbf{102}  & \textbf{102} & \textbf{102} & \textbf{102} & \textbf{102} & \textbf{102} & \textbf{102}  & \textbf{102} & \textbf{102} & \textbf{102} \\
    \cmidrule{1-14}
    FastEnum & & 39 & 39 & 39  & 0 & 0 & 0 & 88 & 88 & 88  & \textbf{102}& \textbf{102} & \textbf{102}\\
    Enum & &  0 & 0 & 0  & 0 & 0 & 0 & 0 & 0 & 0  & 0 & 0 & 0 \\
    Maplex & 8&  0& 0& 0 &0 & 0 & 0 & \textbf{102} & \textbf{102} & \textbf{102} & \textbf{102} & \textbf{102} & \textbf{102} \\
    Base BnB & &  0 & 0 & 0  & 0 & 0 & 0 & \textbf{102} & \textbf{102} & \textbf{102} & \textbf{102} & \textbf{102} & \textbf{102} \\
    Learn BnB & &  \textbf{102} & \textbf{102} & \textbf{102}  & \textbf{102} & \textbf{102} &\textbf{102} & \textbf{102} & \textbf{102} & \textbf{102}  & \textbf{102} & \textbf{102} \textbf{102} \\
    \cmidrule{1-14}
    FastEnum &  & 37 & 37 & 37  & 0& 0 & 0 & 88 & 88 & 88  & 0 &0 & 0 \\
    Enum & &  0 & 0 & 0  & 0 & 0 & 0 & 0 & 0 & 0  & 0 & 0 & 0 \\
    Maplex & 16&  0 & 0 & 0  & 0 & 0 & 0 & \textbf{102} & \textbf{102} & \textbf{102}& \textbf{102} & \textbf{102} & \textbf{102} \\
    Base BnB & &  0 & 0 & 0  & 0 & 0 & 0& \textbf{102} & \textbf{102} & \textbf{102}  & \textbf{102} & \textbf{102} & \textbf{102} \\
    Learn BnB & &  \textbf{102} & \textbf{102} & \textbf{102}  & \textbf{102} & \textbf{102} &\textbf{102}& \textbf{102} & \textbf{102} &\textbf{102} & \textbf{102}& \textbf{102} & \textbf{102} \\
    \cmidrule{1-14}
    FastEnum &  & - & - & -  & 0 & 0 & 0 & 0 & 0 & 0 & 0 & 0 & 0 \\ 
    Enum & & -  & - & -  & 0 &0 & 0 & 0 & 0 & 0  & 0 & 0 & 0 \\
    MaPlex & 32& -  & - &  - &0 & 0 & 0 & 0 & 0 & 0  & \textbf{102} & \textbf{102} & \textbf{102} \\
    Base BnB & &  - &-  & -  & 0 & 0 & 0 & 0 & 0 & 0  & \textbf{102} & \textbf{102} & \textbf{102} \\
    Learn BnB & & -  &-  &  - & \textbf{102} & \textbf{102} & \textbf{102} & \textbf{102} & \textbf{102} & \textbf{102}  & \textbf{102} & \textbf{102} & \textbf{102} \\
    \bottomrule
  \end{tabular}}
  \begin{tablenotes}\footnotesize
    \item[*] *Since when k=32, finding k-plex of lb=30 is trivial, we do not conduct experiments under these parameters.
    \item[*] *For the method \emph{S2V-DQN}, we can not load the whole graph into memory if the graph after prerocessing is to large, so there's no result under certain lb and t.
    \end{tablenotes}
\end{table*}

\subsubsection{Comparison of Basic and Learn Branch and Bound}
To examine the accuracy of the Learn BnB and compare the efficiency between Learn BnB and Basic BnB, we let the two algorithms finish all the searches and then measure their performance. In Table \ref{tab:comparison}, we conduct the experiment on 5 small graphs, the maximize k-plex size detect by the two algorithms are the same, but the Learn BnB can shrink the search time to at least 4 times and up to 80 times compare to Basic BnB. We also measure the accuracy of the bound by Learn BnB decisions. The accuracy means the percent of the bound by Learn BnB do the right deicion(do not bound the feasible solution). The experiment shows that the Learn BnB bound much more than the Basic BnB but did not affect the detecting k-plex size. Since the Basic BnB bound strategy is not a strong bound strategy, the Learn BnB might learn the new bound pattern on the difference constraint space, which perform efficiently.
\begin{table}
  \caption{Comparison of Basic and Learn Branch and Bound}
\label{tab:comparison}
\resizebox{\columnwidth}{!}{%
\def\arraystretch{1.2}
  \begin{tabular}{ccccccc}
    \toprule
    Graph &$|V|$ & $|E|$ &method & time & k-plex size & accuracy\\
    \midrule
aves-sparrowlyon-flock-season3 & 27 & 164 & Basic BnB & 395(s) & 11 & - \\
 &  &   &Learn BnB & 15(s) &11 & 0.83 \\
\cmidrule{1-7}
aves-weaver-social-03 & 42 & 152 & Basic BnB & 2399(s) & 10 & - \\
 &  &   &Learn BnB & 31(s) &10 & 0.99 \\
\cmidrule{1-7}
johnson 8-2-4 & 28 & 210 & Basic BnB & 3412(s) & 5 & - \\
 &  &   &Learn BnB & 744(s) &5 & 0.978 \\
\cmidrule{1-7}
aves-sparrow-social-2009 & 31 & 211 & Basic BnB& 2971(s)  & 11 & - \\
 &  &   &Learn BnB & 149(s)  &11 & 0.901 \\
\cmidrule{1-7}
insecta-beetle-group-c1-period-1 & 30 & 185 & Basic BnB & 2101(s) & 7 & - \\
 &  &   &Learn BnB & 190(s) &7 & 0.993 \\
    \bottomrule
  \end{tabular}}
\end{table}
\subsubsection{Sensitivity Analysis}
After the experiment on the above networks, we want to analyze how algorithms perform under the different values of \(k\) and \(lb\). Figure \ref{fig:k} and Figure \ref{fig:lb} show the average size of detecting k-plex in time limit 300 seconds with different value of \(k\) and \(lb\). In Figure \ref{fig:DBLP-k}, we can find that by the value of k growing, FastEnum and Basic BnB become hard to detect k-plex. For FastEnum, the big value of k enhances the number of possible k-plex, thus FastEnum spends lots of time enumerating the combination extends from the current k-plex with the select vertices, which prevents it from detecting a large k-plex in a short time. For Basic BnB, the growth of k indicated the graph after preprocessing step becomes larger, then increases the computation load of the bound decision. Maplex performance also decreases in k=2 to 16, it might be the same reason as Basic BnB. In Figure \ref{fig:WEB-k}, here only FastEnum decrease with k growing. The two branch and bound based methods seem not influence by k until k=32, but these two methods perform worse than those performed in \emph{DBLP} networks. We surmise that the graph density also influences the performance. If the graph density is high, the impact of k might be diluted. Thus we conduct the experiment to compare graph density following. In contrast, Learn BnB is also influenced by k in Figure \ref{fig:WEB-k} but only to a small extent. The experiment shows our Learn BnB perform more stable in different value of k.

Figure \ref{fig:lb} show all the methods tend to detect a larger k-plew while \(lb\) growing. One reason is that if the network does exist the k-plex is larger or equal than \(lb\), the preprocessing step can efficiently discard the vertices not satisfy the k-lex property, thus the remaining graph \(G^{'}\) is small to make the searching easier. But in the real scenario, we can not assume we can always guess the \(lb\) correctly, the high \(lb\) might make the algorithms can not find a solution at all. So it is important to see if the algorithm works in different \(lb\).  All baselines performance drops while \(lb\) becomes smaller. It seems that no matter the enumeration strategy or the branch and bound strategy both hard to extract the node to extend in the large set of vertices. In contrast, Learn BnB still can find the large k-plex when \(lb\) is low.
\begin{figure}[h]
  \centering
    \subfigure[DBLP]{\includegraphics[width=4.2cm]{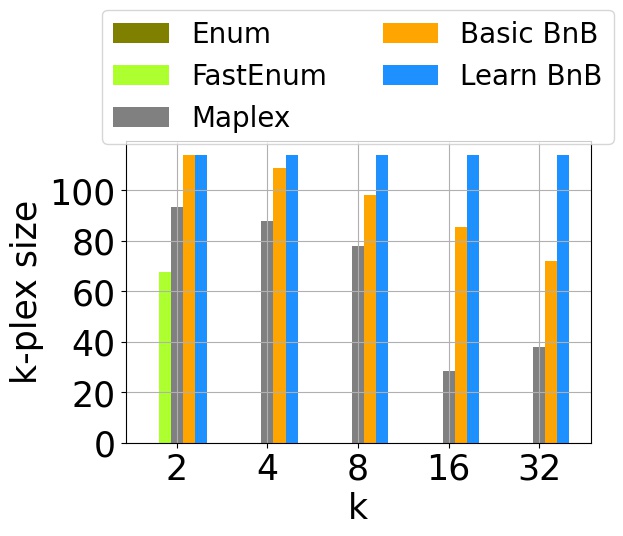}
    \label{fig:DBLP-k}}\subfigure[web-arabic-2005]{\includegraphics[width=4.2cm]{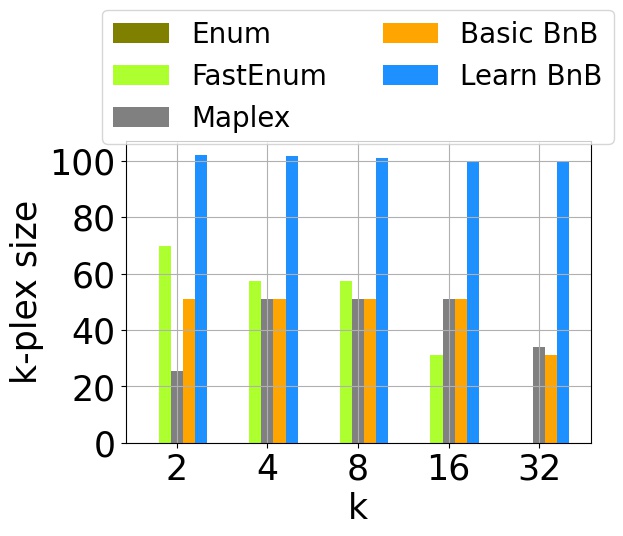}
    \label{fig:WEB-k}}
    \caption{The average size of k-plex with different k at t=300s }
    \label{fig:k}
\end{figure}
\begin{figure}[h]
\centering
\subfigure[DBLP]{\includegraphics[width=4.2cm]{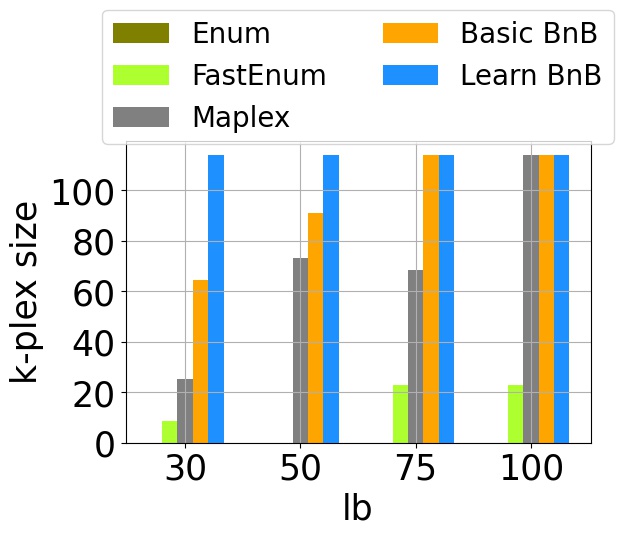}
\label{fig:DBLP-lb}}\subfigure[web-arabic-2005]{\includegraphics[width=4.2cm]{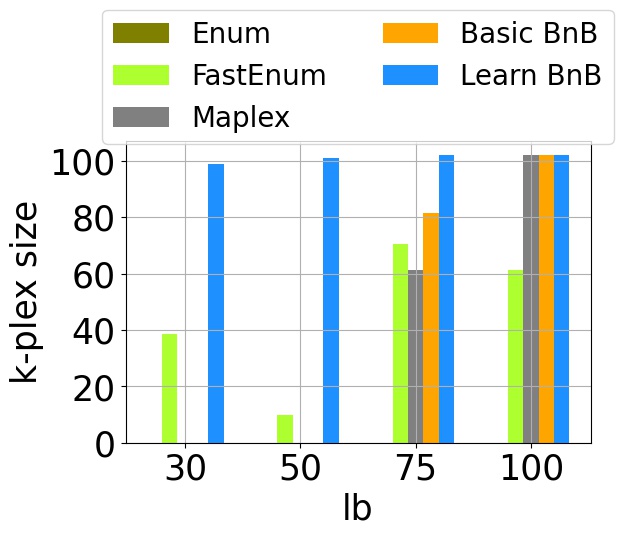}
\label{fig:WEB-lb}}
\caption{The average size of k-plex with different lb }
\label{fig:lb}
\end{figure}

\subsubsection{Graph Properties}
We conduct the experiment on the four collaboration networks to see how graph density affect the algorithms' performance.The result of the collaboration networks is shown in Table \ref{tab:erdos992},\ref{tab:grqc},\ref{tab:condmat},\ref{tab:astroph}. Values in table represent the size of k-plex detected by algorithms under the different parameters.  

See Figure \ref{tab:ca},  the size of \emph{erdos992} and \emph{GrQc} are closed but density of the \emph{GrQc} is 2.5times of \emph{erdos992}. The baselines seems work better in the high density graph \emph{GrQc} while hard to find the solution in \emph{erdos992}. Basic BnB and our Learn BnB finds the k-plex when k=4, lb=10 in \emph{erdos992} while others can not. Than Learn BnB van find the large k-plex even \(lb\) is small in \emph{GrQc}, it seems Learn work stable in these networks.

Graph \emph{Condmat} and \emph{AstroPh} node size are both around 20 thousands, but the average degree of \emph{AstroPh} is 2 to 3 times of \emph{Condmat}. In the sparse graph \emph{Condmat} at k=2, only our Learn BnB finds out the solution on each possible \(lb\); the  Basic BnB still has the obstacle of computing load. EnumFast detects the k-plex whose size is 23 at \(lb\)=20, but it cannot detect any k-plex at \(lb\)=10, it shows that the algorithm does not has stable performance at all parameter, the same as the Maplex algorithm.

The result of graph \emph{AstroPh} is different from others. When k=2, the FastEnum and Learn BnB are lead and perform better than each other at different parameters. Here we again surmise the key point is the first step of the FastEnum can generate the k-plex quickly while the graph has high density. Although the Enum has the same step,  it takes the input graph without pruning, so the average degree is small and makes it difficult to extend the subset with this naive approach. The branch and bound base algorithms Maplex and Basic BnB do not perform well on this network, it can prove again when the graph is large enough and has a high density, it will aggravate the computing load of the bound process. Under this situation, our Learn BnB can still find a solution under size and time constraints. At k=4, the EnumFast loses its advantage and only the Learn BnB generates a good solution.

Based on the observation above, we infer that FastEnum performs well on the high-density graph with s small value of k, but hard to perform its effect on the sparse graph; the Enum is hard to build a solution for the large graph without preprocessing; Maplex performance is not stable and its performance has a relation with different \(lb\) and k; Basic BnB performance decline when the graph size become larger and average degree growing because of the computation load; ours Learn BnB is affected little by the parameters and graph properties, which perform stable than others. It can say that the Learn BnB successfully downgrades the computation load of each bound process and also perform the effect of bound, which makes it can generalize to different properties and conditions.
\begin{table*}
\centering
  \caption{The experimental results of collaboration network Erdos992. The value is the maximize k-plex size found by each methods under different value of \emph{k}, \emph{lb } and \emph{t}. Larger values represent better performance.}
\label{tab:erdos992}
  \resizebox{17cm}{!}{%
\def\arraystretch{1.2}
  \begin{tabular}{cccccccccccccccccccccc}
    \toprule
    \multicolumn{1}{c}{\emph{method}}  &\multicolumn{20}{c}{\emph{k-plex size}} \\
    \midrule
    &\emph{k} &  \multicolumn{10}{c}{2} & \multicolumn{10}{c}{4} \\
    \cmidrule{2-22}
    &\emph{lb} &  \multicolumn{2}{c}{5} & \multicolumn{2}{c}{10} & \multicolumn{2}{c}{20} & \multicolumn{2}{c}{30}  & \multicolumn{2}{c}{40} & \multicolumn{2}{c}{5} & \multicolumn{2}{c}{10} & \multicolumn{2}{c}{20} & \multicolumn{2}{c}{30}  & \multicolumn{2}{c}{40}\\
    \cmidrule{2-22}
     &\emph{time limit} & \multicolumn{1}{c}{10(s)} & \multicolumn{1}{c}{60(s)}  & \multicolumn{1}{c}{10(s)} & \multicolumn{1}{c}{60(s)} &  \multicolumn{1}{c}{10(s)} & \multicolumn{1}{c}{60(s)} &  \multicolumn{1}{c}{10(s)} & \multicolumn{1}{c}{60(s)}  & \multicolumn{1}{c}{10(s)} & \multicolumn{1}{c}{60(s)}  & \multicolumn{1}{c}{10(s)} & \multicolumn{1}{c}{60(s)} &  \multicolumn{1}{c}{10(s)} & \multicolumn{1}{c}{60(s)} &  \multicolumn{1}{c}{10(s)} & \multicolumn{1}{c}{60(s)}  & \multicolumn{1}{c}{10(s)} & \multicolumn{1}{c}{60(s)}  & \multicolumn{1}{c}{10(s)} & \multicolumn{1}{c}{60(s)}   \\
    \cmidrule{2-22}
    FastEnum && 5 & 0 & 0 & 0 & 0  &\textbf{8}  & 0 & 0 & 0 & 0 & 0 & 0 &0 &0 &0 & \textbf{5} & 0 &0 &0 &0 \\
    Enum && 5 & 0 & 0 & 0 & 0  &5  & 0 & 0 & 0 & 0 & 0 & 0 &0 &0 &0 & \textbf{5} & 0 &0 &0 &0 \\
    Maplex && \textbf{8} & 0 & 0 & 0 & 0  &\textbf{8}  & 0 & 0 & 0 & 0 & 0 & 0 &0 &0 &0 & 0 & 0 &0 &0 &0 \\
    Basic BnB && 0 & 0 & 0 & 0 & 0  &0  & 0 & 0 & 0 & 0 & 0 & \textbf{10} &0 &0 &0 & \textbf{5}& \textbf{10} &0 &0 &0 \\
    Learn BnB && 5 & 0 & 0 & 0 & 0  &5  & 0 & 0 & 0 & 0 & \textbf{5} & \textbf{10} &0 &0 &0 & \textbf{5} & \textbf{10} &0 &0 &0 \\
    \bottomrule
  \end{tabular}}
\end{table*}

\begin{table*}
\centering
  \caption{The experimental results of collaboration network GrQc. The value is the maximize k-plex size found by each methods under different value of \emph{k}, \emph{lb } and \emph{t}. Larger values represent better performance.}
\label{tab:grqc}
  \resizebox{17cm}{!}{%
\def\arraystretch{1.2}
  \begin{tabular}{cccccccccccccccccccccc}
    \toprule
    \multicolumn{1}{c}{\emph{method}}  &\multicolumn{20}{c}{\emph{k-plex size}} \\
    \midrule
    &\emph{k} &  \multicolumn{10}{c}{2} & \multicolumn{10}{c}{4} \\
    \cmidrule{2-22}
    &\emph{lb} &  \multicolumn{2}{c}{5} & \multicolumn{2}{c}{10} & \multicolumn{2}{c}{20} & \multicolumn{2}{c}{30}  & \multicolumn{2}{c}{40} & \multicolumn{2}{c}{5} & \multicolumn{2}{c}{10} & \multicolumn{2}{c}{20} & \multicolumn{2}{c}{30}  & \multicolumn{2}{c}{40}\\
    \cmidrule{2-22}
     &\emph{time limit} & \multicolumn{1}{c}{10(s)} & \multicolumn{1}{c}{60(s)}  & \multicolumn{1}{c}{10(s)} & \multicolumn{1}{c}{60(s)} &  \multicolumn{1}{c}{10(s)} & \multicolumn{1}{c}{60(s)} &  \multicolumn{1}{c}{10(s)} & \multicolumn{1}{c}{60(s)}  & \multicolumn{1}{c}{10(s)} & \multicolumn{1}{c}{60(s)}  & \multicolumn{1}{c}{10(s)} & \multicolumn{1}{c}{60(s)} &  \multicolumn{1}{c}{10(s)} & \multicolumn{1}{c}{60(s)} &  \multicolumn{1}{c}{10(s)} & \multicolumn{1}{c}{60(s)}  & \multicolumn{1}{c}{10(s)} & \multicolumn{1}{c}{60(s)}  & \multicolumn{1}{c}{10(s)} & \multicolumn{1}{c}{60(s)}   \\
    \cmidrule{2-22}
    FastEnum && 0 & \textbf{44} & \textbf{44} & \textbf{44} & \textbf{44}  & 0 & \textbf{44} & \textbf{44} & \textbf{44} & \textbf{44} & \textbf{46} &0 &0 &0 &0 &\textbf{46} &0 &0 &0 &0 \\
    Enum && 0 & 0 & 0 & 0 & 0  & 0 & 0 & 0 & 0 & 0 & 0 &0 &0 &0 &0 &0 &0 &0 &0 &0 \\
    Maplex && 0 & 0 & 0 & 0 & 0  & 0 & 0 & 0 & 0 & 0 & 36 &\textbf{46} &\textbf{46} &\textbf{46} &\textbf{46} &\textbf{46} &\textbf{46} &\textbf{46} &\textbf{46} &\textbf{46} \\
    Basic BnB && 7 & \textbf{44} & \textbf{44} & \textbf{44} & \textbf{44}  & \textbf{42} & \textbf{44} & \textbf{44} & \textbf{44} & \textbf{44} & 0 &\textbf{46} &\textbf{46} &\textbf{46} &\textbf{46} &16 &\textbf{46} &\textbf{46} &\textbf{46} &\textbf{46} \\
    Learn BnB && \textbf{44} & \textbf{44} & \textbf{44} & \textbf{44} & \textbf{44}  & \textbf{44} & \textbf{44} & \textbf{44} & \textbf{44} & \textbf{44} & \textbf{46} &\textbf{46} &\textbf{46} &\textbf{46} &\textbf{46} &\textbf{46} &\textbf{46} &\textbf{46} &\textbf{46} &\textbf{46} \\
    \bottomrule
  \end{tabular}}
\end{table*}

\begin{table*}
\centering
  \caption{The experimental results of collaboration network Condmat. The value is the maximize k-plex size found by each methods under different value of \emph{k}, \emph{lb } and \emph{t}. Larger values represent better performance.}
\label{tab:condmat}
  \resizebox{17cm}{!}{%
\def\arraystretch{1.2}
  \begin{tabular}{cccccccccccccccccccccc}
    \toprule
    \multicolumn{1}{c}{\emph{method}}  &\multicolumn{20}{c}{\emph{k-plex size}} \\
    \midrule
    &\emph{k} &  \multicolumn{10}{c}{2} & \multicolumn{10}{c}{4} \\
    \cmidrule{2-22}
    &\emph{lb} &  \multicolumn{2}{c}{5} & \multicolumn{2}{c}{10} & \multicolumn{2}{c}{20} & \multicolumn{2}{c}{30}  & \multicolumn{2}{c}{40} & \multicolumn{2}{c}{5} & \multicolumn{2}{c}{10} & \multicolumn{2}{c}{20} & \multicolumn{2}{c}{30}  & \multicolumn{2}{c}{40}\\
    \cmidrule{2-22}
     &\emph{time limit} & \multicolumn{1}{c}{10(s)} & \multicolumn{1}{c}{60(s)}  & \multicolumn{1}{c}{10(s)} & \multicolumn{1}{c}{60(s)} &  \multicolumn{1}{c}{10(s)} & \multicolumn{1}{c}{60(s)} &  \multicolumn{1}{c}{10(s)} & \multicolumn{1}{c}{60(s)}  & \multicolumn{1}{c}{10(s)} & \multicolumn{1}{c}{60(s)}  & \multicolumn{1}{c}{10(s)} & \multicolumn{1}{c}{60(s)} &  \multicolumn{1}{c}{10(s)} & \multicolumn{1}{c}{60(s)} &  \multicolumn{1}{c}{10(s)} & \multicolumn{1}{c}{60(s)}  & \multicolumn{1}{c}{10(s)} & \multicolumn{1}{c}{60(s)}  & \multicolumn{1}{c}{10(s)} & \multicolumn{1}{c}{60(s)}   \\
    \cmidrule{2-22}
    FastEnum && 0 & 0 & 23 & 0 & 0  & \textbf{7} & 0 & 23 & 0 & 0 & 0 &0 &0 &0 &0 & 0 &0 &0 &0 &0 \\
Enum && 0 & 0 & 0 & 0 & 0  & \textbf{7} & 0 & 0 & 0 & 0 & 0 &0 &0 &0 &0 &0 &0 &0 &0 &0 \\
Maplex && 0 & 0 & 0 & 0 & 0  & 0 & 14 & \textbf{26} & 0 & 0 &0 &0 &0 &0 &0 &0 &\textbf{11} &0 &0 &0 \\
Basic BnB &&0  & 0 & \textbf{26} & 0 & 0  & 0 & 0 & \textbf{26} & 0 & 0 & 0 &0 &0 &0 &0 &0 &0 &0 &0 &0 \\
Learn BnB && \textbf{5} & \textbf{23} & \textbf{26} & 0 & 0  &5  & \textbf{23} & \textbf{26} & 0 & 0 & \textbf{6} & \textbf{11} &0 &0 &0 & \textbf{6} & \textbf{11} &0 &0 &0 \\
    \bottomrule
  \end{tabular}}
\end{table*}

\begin{table*}
\centering
  \caption{The experimental results of collaboration network AstroPh. The value is the maximize k-plex size found by each methods under different value of \emph{k}, \emph{lb } and \emph{t}. Larger values represent better performance.}
\label{tab:astroph}
  \resizebox{17cm}{!}{%
\def\arraystretch{1.2}
  \begin{tabular}{cccccccccccccccccccccc}
    \toprule
    \multicolumn{1}{c}{\emph{method}}  &\multicolumn{20}{c}{\emph{k-plex size}} \\
    \midrule
    &\emph{k} &  \multicolumn{10}{c}{2} & \multicolumn{10}{c}{4} \\
    \cmidrule{2-22}
    &\emph{lb} &  \multicolumn{2}{c}{5} & \multicolumn{2}{c}{10} & \multicolumn{2}{c}{20} & \multicolumn{2}{c}{30}  & \multicolumn{2}{c}{40} & \multicolumn{2}{c}{5} & \multicolumn{2}{c}{10} & \multicolumn{2}{c}{20} & \multicolumn{2}{c}{30}  & \multicolumn{2}{c}{40}\\
    \cmidrule{2-22}
     &\emph{time limit} & \multicolumn{1}{c}{10(s)} & \multicolumn{1}{c}{60(s)}  & \multicolumn{1}{c}{10(s)} & \multicolumn{1}{c}{60(s)} &  \multicolumn{1}{c}{10(s)} & \multicolumn{1}{c}{60(s)} &  \multicolumn{1}{c}{10(s)} & \multicolumn{1}{c}{60(s)}  & \multicolumn{1}{c}{10(s)} & \multicolumn{1}{c}{60(s)}  & \multicolumn{1}{c}{10(s)} & \multicolumn{1}{c}{60(s)} &  \multicolumn{1}{c}{10(s)} & \multicolumn{1}{c}{60(s)} &  \multicolumn{1}{c}{10(s)} & \multicolumn{1}{c}{60(s)}  & \multicolumn{1}{c}{10(s)} & \multicolumn{1}{c}{60(s)}  & \multicolumn{1}{c}{10(s)} & \multicolumn{1}{c}{60(s)}   \\
    \cmidrule{2-22}
    FastEnum && \textbf{20} & 20 & 21 & 44 & \textbf{56}  &\textbf{20}  & 20 & 21 & 44 & 56 & 0 & 0 &0 &0 &0 & 0 & 0 &0 &0 &0 \\
    Enum && 0 & 0 & 0 & 0 & 0  &\textbf{20}  & 20 & 20 & 0 & 0 & 0 & 0 &0 &0 &0 & 0 & 0 &0 &0 &0 \\
    Maplex && 0 & 0 & 0 & 0 & 0  &0  & 0 & 0 & 35 & 0 & 0 & 0 &0 &0 &0 & 0 & 0 &0 &38 &0 \\
    Basic BnB && 0 & 0 & 0 & 0 & 0  &0  & 0 & 0 & 0 & 0 & 0 & 0 &0 &0 &0 & 0 & 0 &0 &0 &0 \\
    Learn BnB && 16 &\textbf{22} & \textbf{50} & \textbf{56} & 53  &16  & \textbf{22} & \textbf{48} & \textbf{56} & 53 & \textbf{39} & \textbf{16} &\textbf{56} &\textbf{41} &\textbf{53} & \textbf{53} & \textbf{16} &\textbf{56} &\textbf{41} &\textbf{53} \\
    \bottomrule
  \end{tabular}}
\end{table*}

\section{Conclusion}
In this paper, we propose the learn to bound framework for detecting k-plex in a short time, we take the concept of constraint learning automated learning the bound strategy, which successfully accelerates the branch and bound algorithm. The experiments show our method generalizes on networks with different properties and also works well under different conditions. The learn to bound strategy can be applied to any branch and bound algorithm with the appropriate framework, which provides a new version of solving the combinatorial optimization problems.



\bibliographystyle{ACM-Reference-Format}
\bibliography{main}

\end{document}